%% file: paper.tex
\documentclass[orivec]{llncs}

\usepackage[noend,linesnumbered]{algorithm2e}
\usepackage{amsmath,amssymb}
\usepackage{stmaryrd}
\usepackage{comment}
\usepackage[usenames,dvipsnames]{color}
\usepackage{listings}
\usepackage{graphicx}
\usepackage{wrapfig}
\usepackage{datetime}
\usepackage{times} 

\newcommand{\clstset}{
\lstset{
  language=C, 
  showspaces=false, 
  showstringspaces=false, 
  showtabs=false, 
  tabsize=2, 
  captionpos=b, 
  mathescape=true, 
  numbers=none,
  basicstyle=\small\ttfamily,
  columns=fullflexible,
  commentstyle=\rmfamily\itshape,
  morekeywords={},
  literate={|->}{$\mapsto$}1
       {exists}{$\exists$}1
       {forall}{$\forall$}1
       {uplus}{$\uplus$}1
       {-*}{$\wand$}1
       {\\/}{$\lor$}1
       {/\\}{$\land$}1
}}

\clstset{}

\newcommand{\partmap}{\rightharpoonup}
\newcommand{\partitive}{\mathcal{P}}
\newcommand{\eqdef}{\triangleq}
\newcommand{\msetfamily}{{\mathcal{M}}}

\newcommand{\ASSERT}{\ensuremath{\mathsf{assert}}}
\newcommand{\ASSUME}{\ensuremath{\mathsf{assume}}}
\newcommand{\Nodes}{\ensuremath{\mathsf{N}}}
\newcommand{\StartNode}{\ensuremath{\mathsf{start}}}
\newcommand{\EndNode}{\ensuremath{\mathsf{end}}}
\newcommand{\Succ}{\ensuremath{\mathsf{succ}}}
\newcommand{\Cmds}{\ensuremath{\mathsf{Cmd}}}
\newcommand{\Cmd}{\ensuremath{\mathsf{cmd}}}

\newcommand{\PAND}{\,\wedge\,}
\newcommand{\POR}{\,\vee\,}
\newcommand{\SPAND}{\,\ast\,}
\newcommand{\LEMPTY}{\ensuremath{\mathsf{emp}}}
\newcommand{\LTRUE}{\ensuremath{\mathsf{true}}}
\newcommand{\LFALSE}{\ensuremath{\mathsf{false}}}
\newcommand{\POINTSTO}[2]{\ensuremath{#1 \mapsto #2}}
\newcommand{\NODE}{\ensuremath{\mathsf{node}}}
\newcommand{\LSEG}{\ensuremath{\mathsf{list}}}
\newcommand{\LSEGless}{\LSEG_{\leq}}
\newcommand{\NIL}{\ensuremath{\mathsf{nil}}}

\newcommand{\SH}{\ensuremath{\mathsf{SH}}}
\newcommand{\CSH}{\ensuremath{\mathsf{CSH}}}
\newcommand{\DSH}{\partitive(\SH{})}
\newcommand{\SHmls}{\SH^{\sf mls}}
\newcommand{\SHrls}{\SH^{\sf rls}}
\newcommand{\SHsls}{\SH^{\sf sls}}
\newcommand{\Spec}{\ensuremath{\mathsf{spec}}}

\newcommand{\EVars}{\ensuremath{\mathsf{EVars}}}
\newcommand{\Project}{\ensuremath{\mathsf{pr}}}
\newcommand{\WeakestPre}{\ensuremath{\mathsf{wp}}}

\newcommand{\Path}{\ensuremath{\mathsf{path}}}
\newcommand{\Level}{\ensuremath{\mathsf{depth}}}
\newcommand{\NodesAt}{\ensuremath{\mathsf{nodes\_at}}}
\newcommand{\Parent}{\ensuremath{\mathsf{parent}}}
\newcommand{\Rearr}[1]{\overline{#1}}
\newcommand{\Inv}{\ensuremath{\mathsf{inv}}}
\newcommand{\Abs}{\ensuremath{\mathsf{abs}}}
\newcommand{\ART}{\ensuremath{\mathsf{ART}}}

\newif\ifcomment

\newcommand{\genericcomment}[2]{
\ifcomment
\begin{center}
\fbox{
\begin{minipage}{4.5in}
{\bf {#2}'s comment:} {\it #1}
\end{minipage}}
\end{center}
\fi}

\newcommand{\mikecomment}[1]{
\genericcomment{#1}{Mike}}

\newcommand{\matkocomment}[1]{
\genericcomment{#1}{Matko}}

\newcommand{\stephencomment}[1]{
\genericcomment{#1}{Stephen}}

\begin{document}
\pagestyle{plain}

\title{Refining Existential Properties \\ in Separation Logic Analyses}


\author{
Matko Botin\v{c}an\inst{1}
\and
Mike Dodds\inst{2}
\and
Stephen Magill\inst{3}
}

\institute{
\email{matko.botincan@gmail.com}
\and
University of York,
\email{mike.dodds@york.ac.uk}
\and
\email{stephen.magill@gmail.com}
}

\maketitle

{\ifcomment
{\bf Compiled:} \currenttime{} \today.
\fi}



\begin{abstract}

In separation logic program analyses, tractability is generally achieved by
restricting invariants to a finite abstract domain.  As this domain cannot
vary, loss of information can cause failure even when verification is possible
in the underlying logic.  In this paper, we propose a CEGAR-like method for
detecting spurious failures and avoiding them by refining the abstract domain.
Our approach is geared towards discovering existential properties, e.g. ``list
contains value x''.  To diagnose failures, we use abduction, a technique for
inferring command preconditions.  Our method works backwards from an error,
identifying necessary information lost by abstraction, and refining the forward
analysis to avoid the error.  We define domains for several classes of
existential properties, and show their effectiveness on case studies adapted
from Redis, Azureus and FreeRTOS.

\end{abstract}

%
%
%
%

\section{Introduction}
\vspace{-0.25em} 

Abstraction is often needed to automatically prove safety properties
of programs, but finding the \emph{right} abstraction can be difficult.
Techniques based on CEGAR ({\bf C}ounter\-{\bf E}xample-{\bf G}uided {\bf A}bstraction
{\bf R}efinement)~\cite{DBLP:conf/cav/ClarkeGJLV00,Kurshan94}
can automatically synthesise an abstraction
that is sufficient for proving a given property. Particularly
successful has been the application of CEGAR to predicate
abstraction~\cite{DBLP:conf/cav/GrafS97}, enabling automated
verification of a wide range of (primarily control-flow driven) safety
properties
\cite{DBLP:conf/spin/BallR01,DBLP:conf/icse/ChakiCGJV03,DBLP:conf/popl/HenzingerJMS02}.

Meanwhile, separation logic has emerged as a useful domain for verifying
shape-based safety properties
\cite{DBLP:conf/aplas/BerdineCO05,DBLP:conf/cav/BerdineCI11,DBLP:journals/jacm/CalcagnoDOY11,DBLP:conf/oopsla/DistefanoP08,DBLP:conf/cav/YangLBCCDO08}.
Its success stems from its ability to compositionally represent heap
operations. The domain of separation logic formulae is infinite, so to ensure
termination, program analyses abstract them by applying a function with a
finite codomain~\cite{DBLP:conf/tacas/DistefanoOY06}.  Although this approach
has proved effective in practice, it does not provide a means to recover from
spurious errors caused by over-abstraction.

This paper proposes a method for automated tuning of abstractions in
separation logic analyses.
Instead of a single abstraction, our method
works with families of abstractions parameterised by multisets,
searching for a parameter that makes the analysis succeed.
Expanding the multiset refines the abstraction, i.e. makes it more precise.

Our method uses a forward analysis that computes a fixpoint using
the current parameterised abstraction, and a backward analysis that
refines this abstraction by expanding the multiset parameter.
To identify the cause of the error and
propagate that information backwards along the counter-examples, we use
abduction, a technique for calculating sufficient preconditions of
program commands~\cite{DBLP:journals/jacm/CalcagnoDOY11}. We use the
difference in symbolic states generated during forward and backward
analysis to select new elements to add to the multiset.

%
%
Our approach focuses on existential properties, where we need to track some
elements of a data structure more precisely than the others.  For example, we
define the domain of ``\emph{lists containing at least particular values}''
(where the multiset parameter specifies the values). In general, our approach
works well for similar existential properties, e.g.  ``\emph{lists containing a
particular subsequence}''.
Existential properties arise e.g.
when verifying that key-value stores preserve values, or in structures that
depend on sentinel nodes.

\mikecomment{Forward pointer to related work on existential properties?}

\matkocomment{I couldn't find any reference dealing with ``existential properties''
like above. There is some work on \emph{existential} reasoning (e.g. Eric's
PLDI'13 paper) but it's about existential paths, i.e. nondeterminism.}

Our approach is complementary to standard separation logic shape analyses. It
adds a new tool to the analysis toolbox, but it is not a general abstraction-refinement
solution.  In particular, \emph{universal} properties such as ``\emph{all list
nodes contain a particular value}'' are not handled.  This is a result of
our analysis structure: when forward analysis fails, we look for portions of
the symbolic state sufficient to avoid the fault, and seek to protect them from
abstraction.  This is intrinsically an existential process.
\vspace{-0.5em} 


\subsection{Related Work}
\vspace{-0.25em} 

%

As in Berdine et al.~\cite{berdine-cav2012} we wish to
automate the process of `tweaking' shape abstractions.
In \cite{berdine-cav2012}, abstract counter-examples are passed to
a SMT solver, which produces concrete counter-example
traces.  These traces determine so-called \textit{doomed states}, conceptually
the same as those singled out for refinement by our procedure.
An advantage of our approach is that we use information
from the failed proof to inform the abstraction refinement
step, rather than exhaustively trying possible refinements as in
\cite{berdine-cav2012}.
This aside, the two approaches are largely complementary:
\cite{berdine-cav2012} focuses on discovering shape refinements, while
our work focuses on data properties.

Our approach operates lazily, but in contrast to lazy
abstraction~\cite{DBLP:conf/popl/HenzingerJMS02}, we do not re-compute
the abstract \emph{post} operator each time we refine the abstraction.
The intermediate formulae we compute during backward analysis can be
seen as interpolants~\cite{DBLP:conf/cav/McMillan06}, but rather than
taking these directly for refining the abstraction, we use them to
select new parameters which have the effect of refining the
abstraction. Such automatic discovery of parameters for parameterised
domains is similar to Naik et al.~\cite{DBLP:conf/popl/NaikYCS12}, however,
instead of analyzing concrete tests, we analyze abstract counter-examples.

Our notion of an abstraction function is similar to
widening~\cite{DBLP:conf/popl/CousotC77}. However, refining the
abstraction with the least upper bound (as
Gulavani and Rajamani~\cite{DBLP:conf/tacas/GulavaniR06}) would not
converge due to the presence of recursive data structures.
\cite{DBLP:conf/sas/ChangRN07} gives a widening for shape domains,
but this widening does not account for data, nor explicitly track
existential properties.
Refinement with an interpolated
widen~\cite{DBLP:conf/tacas/GulavaniCNR08}, while similar to ours, is
also not applicable as we do not work in a complete lattice that is
closed under Craig interpolation. Shape analyses such as
TVLA~\cite{DBLP:journals/toplas/SagivRW02} have been adapted for
abstraction refinement~
\cite{DBLP:conf/cav/BeyerHT06,DBLP:conf/cav/LoginovRS05}, however,
we believe these approaches could not automatically handle
verification of existential properties such as those in \S
\ref{sec:experiments}.

\mikecomment{We say that widening doesn't converge. But we don't necessarily
converge either!}

The refinement process in our approach assumes a parameterised domain
of symbolic heaps which can be refined by augmenting the multiset of
parameters. Compared to predicate abstraction, where the abstract
domain is constructed and refined automatically, in our approach we
first have to hand-craft a parameterised domain. In part this reflects
the intrinsic complexity of shape properties compared to properties
verifiable by standard predicate abstraction.

Several authors have experimented with separation logic domains recording
existential information about stored data,
e.g.~\cite{DBLP:conf/vmcai/Vafeiadis09,Chin2012bag}.  Some of these domains
could be formulated in our multiset-parametric approach, and vice versa.
However, our work differs in that we focus on automating the process of
refining the abstraction.
\vspace{-0.25em} 

\section{Intuitive Description of Our Approach}
\label{sec:example}
\vspace{-0.25em} 

We now illustrate how over-abstraction can cause traditional separation-logic
analyses to fail, and how our approach recovers from such failures.  Our
running example, given in Figure~\ref{fig:example-intraprocedural}, is a simple
instance of the pattern where a value is inserted into an pre-existing
data-structure, the data-structure is further modified, and the program then
assumes the continued presence of the inserted value.  Our code first
constructs a linked list of arbitrary length
(we use `\texttt{*}' for non-deterministic choice). It picks an
arbitrary value for {\tt x}, and creates a node storing this value. It extends
the list with arbitrarily more nodes. Finally, it searches for the node storing
{\tt x} and faults if it is absent.

\begin{figure}[tb]
\vspace{-0.5cm}
\centering
\begin{tabular}{@{}cc}
\begin{minipage}{4cm}
\vspace*{1.5cm}
\begin{lstlisting}[basicstyle=\ttfamily\small]
r = nil;
while (*) {
  r = new Node(r,*);
}
x = *;
r = new Node(r,x);
while (*) {
  r = new Node(r,*);
}
t = r; res = 0;
while(res==0 && t!=nil){
  d = t->data;
  if (d==x) res = 1;
  t = t->next;
}
assert(res==1);
\end{lstlisting}
\end{minipage}
&
{}\hspace{-2cm}
\begin{tabular}{c}
\includegraphics[width=10cm]{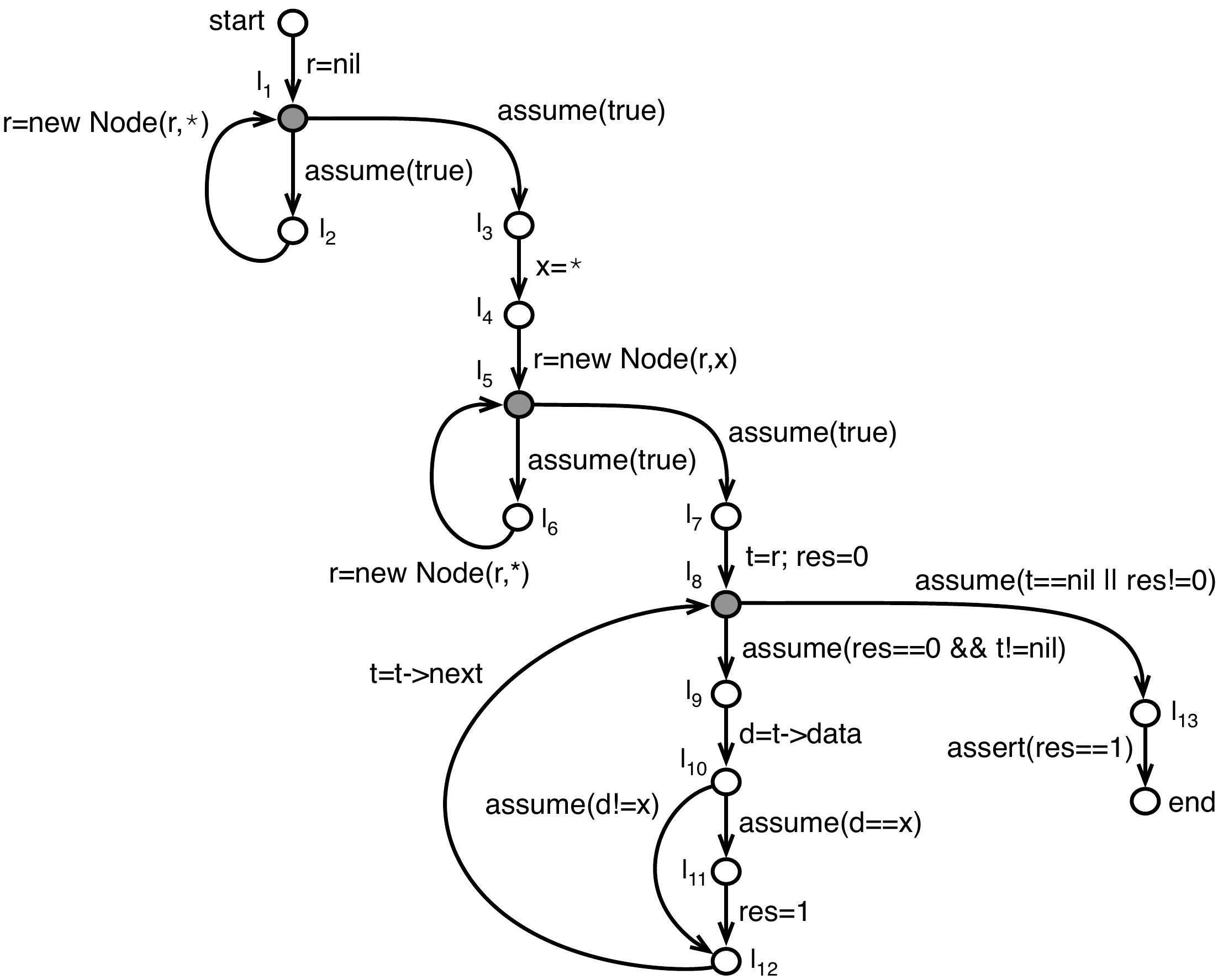}
\end{tabular}
\end{tabular}
\vspace{-1em} 
\caption{Left: running example. Right: associated control-flow graph.
Nodes where abstraction occurs are shaded.}
\label{fig:example-intraprocedural}
\end{figure}

Suppose our abstract domain consists of the predicates $\LEMPTY$,
representing the empty heap, $\NODE(x, y, d)$, representing a linked list
node at address $x$ with next pointer $y$ and data contents $d$, and
$\LSEG(x, y)$, representing a non-empty list segment of unrestricted
length starting at address $x$ and ending with a pointer to $y$.
Nodes and list segments are related by the following recursive definition:
\vspace{-0.5em}
\begin{equation*}
\vspace{-0.5em}
\LSEG(x,y) \quad \eqdef \quad
\NODE(x,y,d')
\vee
(\NODE(x,n',d') * \LSEG(n',y))
\end{equation*}
(Primed variables---$x', y'$, etc.---indicate logical variables that are
existentially quantified).
A traditional analysis, e.g. \cite{DBLP:conf/tacas/DistefanoOY06},
starts with the pre-condition $\LEMPTY$ and propagates symbolic states
over the control-flow graph (right of
Fig.~\ref{fig:example-intraprocedural}). Consider the execution of the
program that adds a single node in the first while loop (node $l_1$)
then adds $\tt x$ to the list, skips the second loop, and then
searches for $\tt x$ (node $l_8$).
Following the two list insertions (node $l_5$) we obtain symbolic state
\vspace{-0.5em}
\begin{equation*}
\vspace{-0.5em}
\NODE({\tt r}, r', {\tt x}) \SPAND \NODE(r', \NIL, d')
\end{equation*}
As is typical, assume the analysis applies the
following abstraction step:
\vspace{-0.5em}
\begin{equation*}
\vspace{-0.5em}
\NODE({\tt r}, r', {\tt x}) \SPAND \NODE(r', \NIL, d') \;\; \leadsto \;\;\LSEG({\tt r}, \NIL)
\end{equation*}
That is, it forgets list length and data values once there are two nodes in
the list. At the head of
the third while-loop (node $l_8$) it unfolds $\LSEG({\tt r}, \NIL)$
back to the single-node case, yielding $\NODE({\tt r}, \NIL, x')$.
Since this state is too weak to show that ${\tt x} = x'$, the path
where ${\tt res}$ is not set to $1$ appears feasible, and the analysis cannot
prove {\tt assert(res==1)}.

\paragraph{Our solution.}
The analysis has failed spuriously because it has abstracted away the
existence of the node containing {\tt x}. We cannot
remove abstraction entirely, and we
also cannot pick a tailored abstraction \emph{a priori}, because the
appropriate abstraction is sensitive to the target program and the
required safety property.
Instead, we work with a
parameterised \emph{family} of abstractions. Starting with the coarsest
abstraction, we modify its parameters based on spurious
failures, automatically
tailoring the abstraction to the property we want to prove.

For our example, we augment the domain with a family of predicates
$\LSEG(\_,\_, {\{ {\tt d} \}})$, representing a list where at least one
node holds the value {\tt d} (this domain is
defined in \S \ref{sec:mls}). Upon failing to prove the program,
our backwards analysis looks for extensions of
symbolic states that \emph{would} satisfy
{\tt assert(res==1)}, and so avoid failure. Technically, this is
achieved by posing successive abduction queries
along the counter-example path. If an extension is found, then the
difference between the formulae from forward and backward analysis
identifies the cause of the spurious failure. In our example, the
analysis infers that the failure was due to the abstraction of the node
storing ${\tt x}$. We refine the abstraction so
nodes containing ${\tt x}$ are rewritten to $\LSEG(\_,\_, {\{ {\tt x} \}})$,
``remembering'' the existence of {\tt x}. This
suffices to prove the program correct.
\vspace{-0.25em} 

\section{Analysis Structure}
\label{sec:background}
\vspace{-0.25em} 



\paragraph{Symbolic heaps.}
A symbolic heap $\Delta$ is a formula of the form $\Pi \PAND \Sigma$
where $\Pi$ (the pure part) and $\Sigma$ (the spatial part) are defined
by:\vspace{-0.5em}
\begin{eqnarray*}
\vspace{-0.5em}
\Pi & ::= & \LTRUE \, \mid \LFALSE \, \mid \, e = e \, \mid \, e \neq e \,
\mid  \, p(\bar{e}) \, \mid \, \Pi \PAND \Pi \\
\Sigma & ::= & \LEMPTY \, \mid \, s(\bar{e})
\, \mid \, \Sigma \SPAND \Sigma
\end{eqnarray*}
\vspace{-0.25em}
Here $e$ ranges over (heap-independent) expressions (built over program
and logical variables), $p(\bar{e})$ over pure predicates and
$s(\bar{e})$ over spatial predicates. Logical variables are
(implicitly) existentially quantified; the set of all such variables in
$\Delta$ is denoted by $\EVars(\Delta)$. $\Sigma_1 * \Sigma_2$ holds
if the state can be split into two parts with disjoint domains, one
satisfying $\Sigma_1$ and the other $\Sigma_2$. A disjunctive symbolic
heap is obtained by combining symbolic heaps (both the pure and spatial
part) with disjunction. We identify a disjunctive heap with the set of
its disjuncts, and also denote such heaps with $\Delta$. The set of all
consistent symbolic (resp. disjunctive) heaps is denoted by $\SH{}$
(resp.~$\DSH{}$).

\paragraph{Abstract domain.}
Our abstract domain is the join-semilattice $(\DSH^{\top}, \vdash,
\sqcup, \top)$, where $\DSH^{\top} \triangleq \DSH{} \cup \{ \top \}$,
the partial order is given by the entailment relation $\vdash$, the
join $\sqcup$ is disjunction, and the top element, $\top$,
represents error.

We assume a sound theorem prover that can deal with entailments between
symbolic heaps, frame inference, and abduction queries (square brackets
denote the computed portion of the entailment):
\vspace{-0.5em}
\begin{itemize}
\item $\Delta_{1} \vdash \Delta_{2} \SPAND [\Delta_F]$ (frame
inference): given $\Delta_{1}$ and $\Delta_{2}$, find the frame
$\Delta_F$ such that $\Delta_{1} \vdash \Delta_{2} \SPAND \Delta_F$
holds;
\item $\Delta_{1} \SPAND [\Delta_A] \vdash \Delta_{2}$ (abduction):
given $\Delta_{1}$ and $\Delta_{2}$, find the `missing' assumption
$\Delta_A$ such that $\Delta_{1} \SPAND \Delta_A \vdash
\Delta_{2}$ holds.
\end{itemize}


\paragraph{Specifications and programs.}
We assume that each atomic command $c \in \Cmds{}$ is associated with a
specification $\{ P \} \, c \, \{ Q \}$, consisting of a precondition
$P$ and a postcondition $Q$ in $\SH{}$ (in fact, our case studies use
specifications expressed by using points-to and (dis)equalities only).
We define $\ASSUME(e) \triangleq \{ {\sf true} \} \, \{ e \}$ and
$\ASSERT(e) \triangleq \{ e \} \, \{ e \}$. Specifications are
interpreted using standard partial correctness: $\{ P \} \, c \, \{ Q
\}$ holds iff when executing $c$ from a state satisfying $P$, $c$ does
not fault, and if it terminates then the resulting state satisfies $Q$.
As is standard in separation logic, we also assume specifications are
\emph{tight}: $c$ will not access any resources outside of the ones
described in $P$.

We represent programs using a variant of intra-procedural control-flow
graphs \cite{DBLP:conf/popl/HenzingerJMS02} over the set of atomic
commands $\Cmds{}$. A CFG consists of a set of nodes $\Nodes{}$
containing distinguished starting and ending nodes $\StartNode{},
\EndNode{} \in \Nodes{}$, and functions, $\Succ \colon \Nodes{} \to
\partitive(\Nodes{})$  and $\Cmd \colon \Nodes{} \times \Nodes{}
\partmap{} \Cmds{}$, representing node successors and edge labels. All
nodes either have a single successor, or all outgoing edges are
labelled with command $\ASSUME(e)$ for the condition $e$ that must hold
for that edge to be taken.

\paragraph{Forward and backward transfer.}
We define the abstract forward semantics of each atomic command $c$ by
a function $\llbracket c \rrbracket \colon \SH{} \to \DSH^{\top}$. The
function $\llbracket c \rrbracket$, fusing together rearrangement
(materialisation) and symbolic execution
\cite{DBLP:journals/toplas/SagivRW02,DBLP:conf/aplas/BerdineCO05,DBLP:conf/tacas/DistefanoOY06,DBLP:journals/jacm/CalcagnoDOY11},
is defined using the \emph{frame rule}, which allows any triple
$\{ P \} \, c \, \{ Q \}$ to be extended by an arbitrary frame $\Delta_F$
that is not modified by $c$:
\vspace{-0.5em} 
$$
\llbracket c \rrbracket (\Delta) \quad \eqdef \quad
\left\{
\begin{array}{l@{\;\;\;}l}
\top & \textrm{if } \nexists \Delta_F \ldotp \, \Delta \vdash P \SPAND \Delta_F
\\[0.2em]
\{ Q \SPAND \Delta_F \mid \Delta \vdash P \SPAND \Delta_F \} & \textrm{otherwise.}
\end{array}
\right.
$$
When there is no $\Delta_F$ such that $\Delta \vdash P \SPAND
\Delta_F$, the current heap $\Delta$ does not satisfy the
precondition $P$ of the command, and so execution may result in an
error. We assume that the prover filters out inconsistent heaps. Lifting
disjunctions to sets on the left-hand side is justified by
the disjunction rule of Hoare logic.
We lift $\llbracket c \rrbracket$ to a forward transfer function
$\DSH^{\top} \to \DSH^{\top}$ by mapping $\top$ to $\top$ and a set of
symbolic heaps to the join of their $\llbracket c \rrbracket$-images.

We use abduction to transfer symbolic heaps backwards: given a
specification $\{ P \} \, c \, \{ Q \}$ and disjunctive symbolic heap
$\Delta$, if $\Delta_A$ is such that $Q \SPAND \Delta_A \vdash \Delta$
then $\{ P \SPAND \Delta_A \} \, c \, \{ \Delta \}$, i.e., we can
``push'' $\Delta$ backwards over $c$ to obtain $P \SPAND \Delta_A$ as a
pre-state. This gives rise to a \emph{backward} transfer function
$\llbracket c \rrbracket^{\leftarrow} \colon \DSH{} \to \DSH{}$ defined
by:
\vspace{-0.5em} 
$$
\llbracket c \rrbracket^{\leftarrow} (\Delta)
\quad \eqdef \quad
\textit{choose}(
\{ P \SPAND \Delta_A \mid Q \SPAND \Delta_A \vdash \Delta \})
$$
The heuristic function $\textit{choose}(-)$ selects a `good'
abduction solution (there can be many, e.g. a trivial one,
$\LFALSE$). For some fragments best
solutions are possible: e.g. the disjunctive points-to fragment with
(dis)equalities~\cite{DBLP:journals/jacm/CalcagnoDOY11}, a variation of
which we use in our backward analysis. Along $\ASSUME$-edges we have
$\llbracket \ASSUME(e) \rrbracket^{\leftarrow}(\Delta) \, = \,
\WeakestPre(\ASSUME(e), \Delta) \, = \, \neg e \POR \Delta$.
\vspace{-0.5em} 


\subsection{Forward Analysis, Abstraction Function, and Parametricity}
\label{sec:forward-analysis}
\vspace{-0.25em} 

Forward analysis attempts to compute an inductive invariant $\Nodes{}
\to \DSH^{\top}$. It gradually weakens the strongest property by
propagating symbolic heaps along CFG edges using the forward transfer,
and joining the obtained $\llbracket c \rrbracket$-images at each CFG
node. Since our abstract domain is infinite, and
transfer functions are not necessarily monotone, forward propagation alone may
not reach a fixpoint, or even converge towards one.

We call a pair $(\SH{}, \Abs)$ an \emph{analysis}.
To ensure termination, propagated symbolic heaps are abstracted
into a finite set, and the propagation process is made
inflationary.\footnote{A function $f \colon (A, \sqsubseteq) \to (A,
\sqsubseteq)$ is \emph{inflationary} if for every $a$, we have $a
\sqsubseteq f(a)$.}
Abstraction is realised by a function $\Abs \colon \SH{} \to \CSH{}$
whose codomain is a finite subset $\CSH{}$ of $\SH{}$. At each
step,
$\Abs$ replaces the propagated symbolic heap with a logically weaker one
in $\CSH{} \cup \{ \top \}$. We require $\Abs$ to be inflationary,
i.e., that it soundly over-approximates symbolic heaps with respect to
$\vdash$. Making the propagation inflationary means that instead of
computing the (least) fixed-point of the functional $\Phi \colon
(\Nodes{} \to \DSH^{\top}) \to (\Nodes{} \to \DSH^{\top})$, we compute
the inflationary fixed-point of the functional $X \mapsto X \cup
\Phi(X)$.



\begin{definition}[analysis comparison]
Let $\Abs \colon \SH{} \to \CSH{}$ and $\Abs' \colon \SH{} \to \CSH{}'$
be abstraction functions. We say that $\Abs'$ refines $\Abs$, written
$\Abs \preceq \Abs'$, if $\CSH{} \subseteq \CSH{}'$ and for every
$\Delta \in \SH{}$, $\Abs'(\Delta) \vdash \Abs(\Delta)$. We say that
$(\SH{}, \Abs')$ is more precise than $(\SH{}, \Abs)$ if $\Abs
\preceq \Abs'$.
\end{definition}


In \S \ref{sec:abs_ref_details} and \S \ref{sec:list-multiset-refinement} we
introduce families of analyses whose abstraction functions are parameterised by
a multiset (such analyses are \emph{parametric} in the sense of
\cite{DBLP:conf/popl/NaikYCS12}). For any such family $(\SH{}, \Abs_S)_{S \in
\msetfamily}$, where $\msetfamily$ is some family of multisets and $\Abs_S
\colon \SH \to \CSH_S$, we require that if $S \subseteq S'$ then $\Abs_{S}
\preceq \Abs_{S'}$.
\vspace{-0.5em} 


\subsection{Forward-Backward Abstraction Refinement Algorithm}
\label{sec:abs_ref_details}
\vspace{-0.25em} 

We now define an intra-procedural version of our analysis formally (we believe
it could be made inter-procedural without difficulty -- see
\S\ref{sec:abs_ref_properties}).  Let $(\SH{},\Abs_S)_{S \in \msetfamily}$ be a
family of analyses parameterised by a multiset. Our method for abstraction
refinement starts with the analysis $(\SH{},\Abs_{\emptyset})$, and iteratively
refines the abstraction by adding terms to the multiset $S$. The goal is to
eventually obtain $S$ such that using the analysis $(\SH{},\Abs_S)$ we can
compute a sufficient inductive invariant.

\SetFuncSty{textbf}

\paragraph{Forward analysis.}
\begin{algorithm}[tb]
\SetKwFunction{KwRefine}{Refine}
\SetKwFunction{KwRefineAbstraction}{SelectSymbols}
$S := \emptyset$; \ $t_0 := (\StartNode{}, \LEMPTY{})$; \ $k := 0$; \ $E_{\mathcal{T}} = \emptyset$; \ $\mathcal{T} =
(\{t_0\}, E_\mathcal{T}, t_0)$\;
\While{$\NodesAt_{\mathcal{T}}(k) \neq \emptyset$}{\label{fig:Fixpoint:main-loop}
  \ForEach{$t = (n, \Delta) \in \NodesAt_{\mathcal{T}}(k)$}{
    \ForEach{$n' \in \Succ(n)$}{
      $\mathcal{D}' := \llbracket \Cmd(n, n') \rrbracket (\Delta)$\;
      \eIf{$\mathcal{D}' = \top$}{
        Add $(t, (n', \top))$ to $E_{\mathcal{T}}$\;
        $k, S := \text{\KwRefine{$(n', \top), S$}}$\;
        Break to the outermost while-loop\;
      }{
        \ForEach{$\Delta' \in \mathcal{D}$}{
          $\Delta_{\sf abs}' := \Abs_S(\Delta')$\;
          \If{$\Delta_{\sf abs}' \not\sqsubseteq \Inv_{\mathcal{T}}(n')$}{
            $\mathcal{T}_{\sf old} := \mathcal{T}$\;
            Add $(t, (n', \Delta_{\sf abs}'))$ to $E_{\mathcal{T}}$\;
            $\Inv_{\mathcal{T}} := \Inv_{\mathcal{T}_{\sf old}} [n' \mapsto \Inv_{\mathcal{T}_{\sf old}}(n') \sqcup \Delta_{\sf abs}']$\;
          }
        }
      }
    }
  }
  $k := k+1$\;
}
\medskip
\caption{Forward analysis with abstraction refinement.}
\label{fig:Fixpoint-Abstraction-Refinement}
\end{algorithm}

Algorithm~\ref{fig:Fixpoint-Abstraction-Refinement} shows a forward
analysis from \S \ref{sec:forward-analysis} extended with abstraction
refinement. The algorithm computes a fixpoint by constructing an
abstract reachability tree (ART). An ART is a tree $\mathcal{T} = (T,
E, t_0) \in \ART{}$ where $T$ is the set of nodes, $E$ the set of edges
and $t_0$ the root node. We write $E_{\mathcal{T}}$ to refer to the set
of edges associated with a particular ART $\mathcal{T}$. Nodes in $T$
are of the form $(n, \Delta) \in \Nodes{} \times \SH{}$ and represent
the abstract states visited during the fixpoint computation. We use the
following functions to deal with the ART: $\Parent_{\mathcal{T}} \colon
T \setminus \{ t_0 \} \to T$ returning the unique parent of a node,
$\Level_{\mathcal{T}} \colon T \to \mathbb{N}_0$ returning the length
of the path from $t_0$ to $t$, and $\NodesAt_{\mathcal{T}} \colon
\mathbb{N}_0 \to \partitive(T)$ returning the set of all nodes at the
given depth. For $\mathcal{T} = (T,E,t_0)$ and $\mathcal{T'} =
(T',E',t_0')$, we write $\mathcal{T} \subseteq \mathcal{T'}$ to
indicate that $\mathcal{T}$ is a subtree of $\mathcal{T'}$, i.e., that
$T \subseteq T'$, and $E \subseteq E'$, and $t_0 = t_0'$.  We write
$\Cmd(n, n')$ to represent the command labelling the edge between nodes
$n$ and $n'$ and $\Spec(n, n')$ for the corresponding specification.

The algorithm iteratively propagates $\llbracket \cdot
\rrbracket$-images of previously-computed abstract states along CFG
edges, applies abstraction if the result is consistent, and joins each
newly computed state with the previously-computed states at the same
node. We store the invariant computed at each step using a map $\Inv{}
\colon \ART{} \to (\Nodes{} \to \DSH^{\top})$, reflecting the fact that
the invariant at a control point can be recovered from the node labels
of the ART. If we have $\NodesAt(k) = \emptyset$ for the current depth
$k$, then we have successfully computed an inductive invariant without
reaching an error.

Suppose at some point the transfer function returns $\top$, i.e., the
forward analysis fails to prove a property (e.g., a pure assertion or a
memory safety pre-condition of a heap-manipulating command). This can
happen due to either a true violation of the property, or a spurious
error caused by losing too much information somewhere along the
analysis. The algorithm then invokes Algorithm~\ref{fig:Refinement},
\KwRefine, to check for feasibility of the error and, if it is
spurious, to refine the abstraction.

\begin{algorithm}[t]
\DontPrintSemicolon
\KwRefine{$t^{\top} \colon T$, $S \colon \msetfamily$}$\colon
   N_0 \times \msetfamily$\;
\PrintSemicolon
\Begin{
  $k := \Level(t^{\top}) - 1$\;
  $t_{\sf curr} := t^{\top}; \; \; t_{\sf prev} := \Parent(t^{\top})$\;
  $\{ P \} \, \_ \, \{ \_ \} := \Spec(t_{\sf prev}.n, t_{\sf curr}.n)$\;
  Solve $\Rearr{t_{\sf prev}.\Delta} \SPAND [\Delta_A] \vdash P \SPAND \LTRUE$\; \label{line:abd_error}
  $\Delta_{\sf prev}' := \Rearr{t_{\sf prev}.\Delta} \SPAND \Delta_A$\;
  $\Path_{\sf fwd} := t_{\sf prev}.\Delta; \;\; \Path_{\sf bwd} := \Delta_{\sf prev}'$\;
  \While{$k>0$}{
    $t_{\sf curr} := t_{\sf prev}; \; \; t_{\sf prev} := \Parent(t_{\sf curr}); \; \; \Delta_{\sf curr}' := \Delta_{\sf prev}'$\;
    $\Delta_{\sf prev}' := \llbracket \Cmd(t_{\sf prev}.n, t_{\sf curr}.n) \rrbracket^{\leftarrow}(\Delta_{\sf curr}')$\; \label{line:back_exec}
    $\Path_{\sf fwd} := t_{\sf curr}.\Delta \cdot \Path_{\sf fwd}; \;\;
       \Path_{\sf bwd} := \Delta_{\sf curr}' \cdot \Path_{\sf bwd}$\;
    \If{$t_{\sf prev}.\Delta \vdash \Delta_{\sf prev}'$}{ \label{line:refine_abs}
      $S := S \cup \text{\KwRefineAbstraction{$\Path_{\sf fwd}, \Path_{\sf bwd}$}}$\; \label{line:call_selectsymbol}
      Delete $t_{\sf curr}$-subtree of $\mathcal{T}$\;\label{fig:Refinement:delete}
      \Return{$k-1, S$}\;
    }
    $k := k - 1$\;
  }
  throw ``error''\; \label{line:exit_error}
}
\medskip
\caption{Backward analysis of counter-example by abduction.}
\label{fig:Refinement}
\end{algorithm}

\paragraph{Backward analysis.}
Algorithm~\ref{fig:Refinement}, \KwRefine, operates by backward
analysis of abstract counter-examples. Rather than using weakest
preconditions as in CEGAR, \KwRefine uses abduction to propagate
formulae backwards along an abstract counter-example and check its
feasibility. Once a point in the path is found where forward analysis
agrees with the backward analysis, the mismatch between the symbolic
heaps from forward and backward analyses is used to update the
multiset $S$ that determines the abstraction.

\begin{figure}[t]
\centering
{\small
$$
\begin{array}{@{}c@{\;\;}c@{}}
\begin{array}{@{}rl}
t_{0} \colon & (\StartNode, \LEMPTY)\\
t_{1} \colon & (l_{1}, {\tt r} = \NIL)\\
t_{2} \colon & (l_{2}, {\tt r} = \NIL)\\
t_{3} \colon & (l_{1}, \NODE({\tt r}, \NIL, \_))\\
t_{4} \colon & (l_{3}, \NODE({\tt r}, \NIL, \_))\\
t_{5} \colon & (l_{4}, \NODE({\tt r}, \NIL, \_))\\
t_{6} \colon & (l_{5}, \LSEG({\tt r}, \NIL))\\
t_{7} \colon & (l_{7}, \LSEG({\tt r}, \NIL))\\
\end{array}
&
\begin{array}{rl@{}}
t_{8} \colon & (l_{8}, \LSEG({\tt r}, \NIL) \PAND {\tt t} = {\tt r} \PAND {\tt res} = 0)\\
t_{9} \colon & (l_{9}, \LSEG({\tt r}, \NIL) \PAND {\tt t} = {\tt r} \PAND {\tt res} = 0)\\
t_{10}\colon  & (l_{10}, \NODE({\tt r}, \NIL, d') \PAND {\tt t} = {\tt r} \PAND {\tt res} = 0 \PAND {\tt d} = d')\\
t_{11}\colon  & (l_{12}, \NODE({\tt r}, \NIL, d') \PAND {\tt t} = {\tt r} \PAND {\tt res} = 0 \PAND {\tt d} = d' \PAND {\tt d} \neq {\tt x})\\
t_{12}\colon  & (l_{8}, \NODE({\tt r}, \NIL, d') \PAND {\tt t} = \NIL \PAND \hfill {\tt res} = 0 \PAND {\tt d} = d' \PAND {\tt d} \neq {\tt x})\\
t_{13}\colon  & (l_{13}, \NODE({\tt r}, \NIL, d') \PAND {\tt t} = \NIL \PAND {\tt res} = 0)\\
t_{14}\colon  & (\EndNode, \top)
\\
\\
\end{array}
\end{array}
$$
}
\vspace{-2em} 
\caption{Abstract counter-example for the running example (\S \ref{sec:example}).}
\label{fig:abstract-counter-example}
\end{figure}

\begin{definition}
An \emph{abstract counter-example} is a sequence $(n_0, \Delta_0) \ldots
(n_k, \Delta_k)$ with:
\vspace{-0.5em} 
\begin{itemize}
\item $n_0  = \StartNode{}$ and for all $0 < i \leq k$, $n_i \in \Succ(n_{i-1})$;
\item $\Delta_0 = \LEMPTY{}$, for all $0 < i \leq k$, $\Delta_i \in
\Abs_S(\llbracket \Cmd(n_{i-1}, n_{i}) \rrbracket (\Delta_{i-1}))$ and
$\Delta_k = \top$.
\end{itemize}
\end{definition}
Figure \ref{fig:abstract-counter-example} shows the abstract
counter-example for the error discussed in \S \ref{sec:example}. This
is the sequence of symbolic heaps computed by the analysis on its way
to the error.  This counter-example covers the case where the list
contains just one node. The error results from over-abstraction, which
has erased the information that this node contains the value 0 (this
can be seen in the last non-error state, $t_{13}$).

\KwRefine begins by finding a resource or pure assumption sufficient
to avoid the terminal error in the counter-example. Let $(n_0,
\Delta_0) \ldots (n_k, \top)$ be an abstract counter-example with
$\Cmd(n_{k-1}, n_{k}) = \{ P_k \} \, c_k \, \{ Q_k \}$. Since
$\llbracket c_k \rrbracket (\Delta_{k-1}) = \top$, $\Delta_{k-1}$
misses some assumption required to satisfy $P_k$. To find this,
\KwRefine solves the following abduction query
(line~\ref{line:abd_error})---here $\Rearr{\Delta_{k-1}}$ is a
rearrangement of $\Delta_{k-1}$, for example to expose particular
memory cells:
\vspace{-0.5em} 
$$
\Rearr{\Delta_{k-1}} \SPAND [\Delta_A] \vdash P_k \SPAND \LTRUE.
$$
The resulting symbolic heap $\Delta_A$ expresses resources
or assumptions that, in combination with $\Rearr{\Delta_{k-1}}$,
suffice to guarantee successful execution of $c_k$.
If $\Delta_A$ is $\LFALSE$ then $\Rearr{\Delta_{k-1}} \SPAND
[\Delta_A]$ is inconsistent; if this happens then the analysis will
have to find a refinement under which $(n_{k-1}, \Delta_{k-1})$ can be
proved to be unreachable.


Letting $\Delta_{k-1}' := \Rearr{\Delta_{k-1}} \SPAND \Delta_A$,
\KwRefine computes a sufficient resource for the preceding state
(line~\ref{line:back_exec}):
\vspace{-0.5em} 
$$
\Delta_{k-2}' :=
\llbracket \Cmd(n_{k-2}, n_{k-1}) \rrbracket^{\leftarrow}
(\Delta_{k-1}').
$$
If $\Delta_{k-2} \vdash \Delta_{k-2}'$ then in the step from
$\Delta_{k-2}$ to $\Delta_{k-1}$ a loss of precision has happened, and
we use the additional information in $\Delta_{k-1}'$ to refine the
abstraction (line \ref{line:refine_abs}). Otherwise, we continue
pushing backwards, and generate $\Delta'_{k-3}$, $\Delta'_{k-4}$, etc.

Eventually, \KwRefine either halts with $\Delta_{i} \vdash \Delta_{i}'$
for some $i \geq 0$, or in the last step obtains $\Delta_0 \not \vdash
\Delta_0'$. In the former case, \KwRefine invokes the procedure
\KwRefineAbstraction, passing it the forward and the backward sequence
of symbolic heaps leading to the error
(line~\ref{line:call_selectsymbol}). The symbols it generates are added
to the multiset $S$, refining the abstraction. In the latter case, we
did not find a point for refining the abstraction, so \KwRefine reports
a (still possibly spurious) error (line~\ref{line:exit_error}). Note
that in this case the computed heap $\Delta_0'$ is a sufficient
pre-condition to avoid this particular abstract counter-example.

If \KwRefine calls \KwRefineAbstraction to update the abstraction, it
discards the current node and all its descendants from the ART (line
\ref{fig:Refinement:delete}). The ART below the refinement point will
be recomputed in subsequent iterations using (possibly) stronger
invariants.

\begin{theorem}[Soundness]
If the algorithm terminates without throwing an error, the computed map
$\Inv_{\mathcal{T}}$ is an
inductive invariant not containing $\top$.
\end{theorem}
\begin{proof}
Refinement in
Alg.~\ref{fig:Fixpoint-Abstraction-Refinement} is achieved by
augmenting $S$ with new elements selected by
\KwRefineAbstraction. Since $\Abs_{S} \preceq \Abs_{S'}$ for $S
\subseteq S'$, this is immediately sound.
\qed
\end{proof}

\paragraph{Refinement heuristic.}

\KwRefineAbstraction stands for some heuristic function which refines the
abstraction.  It takes two sequences, $\Path_{\sf fwd}$ and $\Path_{\sf bwd}$:
the former is a path taken by the forward analysis from the $i$-th node of the
counter-example to the error node (such that $\Delta_{i} \vdash \Delta_{i}'$ in
Alg.~\ref{fig:abstract-counter-example}), while the latter is a path sufficient
to avoid the error. \KwRefineAbstraction tries to identify symbols present in
the error-avoiding path that have been lost in the forward, overly-abstracted
path.  Conceptually, \KwRefineAbstraction can be seen as a simpler analogue of
the predicate discovery heuristics
\cite{DBLP:conf/spin/BallR01,DBLP:conf/popl/HenzingerJMS02}
(it synthesizes
only symbolic constants rather than predicates).

The heuristic in our implementation works by examining the syntactic structure
of formulae $\Delta$ and $\Delta'$ for which $\Delta \vdash \Delta'$ has been
established.
The heuristic starts by identifying congruence classes of terms occurring
in both formulae and building a tree of equalities between program variables in
each congruence class. Our separation logic prover preserves the common
syntactic parts of $\Delta$ and $\Delta'$ by explicitly recording
substitutions, ensuring we can recover a mapping between common variables
occurring in both formulae.
The heuristic then exhaustively traverses equalities in the congruence classes
for each term of $\Delta'$, and checks whether equalities can be
used to strengthen $\Delta$ without making it inconsistent. Intuitively,
because these equalities are mentioned in the calculated sufficient resource,
they will likely be significant for program correctness.  The variables in
identified equalities are then used to strengthen the abstraction.
We found this heuristic worked well in our case studies (see \S
\ref{sec:experiments}).

\mikecomment{Every time we've submitted this paper we've got complaints about
the lack of detail on \KwRefineAbstraction. Write a pseudocode implementation?}

\paragraph{Running example revisited.}


In \S \ref{sec:example} we saw a spurious error caused by
over-abstracting values in the list. To fix this, we augmented the
domain with predicates $\LSEG(\_,\_, {\{ {\tt d} \}})$, representing a
list that has at least one node with value ${\tt d}$. We now show the
refinement step in this domain.
The backward analysis begins by solving the abduction query
$$(\NODE({\tt r}, \NIL, d') \PAND {\tt x} = d \PAND {\tt t} = \NIL \PAND {\tt res} = 0) \SPAND [\Delta'_{13}] \; \vdash \; {\tt res} = 1 \SPAND [\_]$$
This yields $\Delta'_{13} = \LFALSE$ as the only solution. The
analysis then generates the following sequence of symbolic heaps (we omit some
for brevity).
Compare with the abstract counter example in
Fig.~\ref{fig:abstract-counter-example}; here $\Delta_i'$ corresponds to node
$t_i$):
\small
\begin{eqnarray*}
\Delta_{12}' & = & ({\tt t} \neq \NIL \PAND {\tt res} = 0 \PAND \LTRUE)\\
\Delta_{10}' & = & (({\tt d} = {\tt x} \PAND \LTRUE) \POR (t' \neq \NIL \PAND {\tt res} = 0 \PAND \NODE({\tt t}, t', d') \SPAND \LTRUE))\\
\Delta_{9}' & = & ((\NODE({\tt t}, t', {\tt x}) \SPAND \LTRUE) \POR (t' \neq \NIL \PAND {\tt res} = 0 \PAND \NODE({\tt t}, t', d') \SPAND \LTRUE))\\
\Delta_{8}' & = & ( (({\tt res} \neq 0 \POR {\tt t} = \NIL) \PAND \LTRUE) \POR (\NODE({\tt t}, t', {\tt x}) \SPAND \LTRUE) \POR
\\[-0.2em]
& & (t' \neq \NIL \PAND {\tt res} = 0 \PAND \NODE({\tt t}, t', d') \SPAND \LTRUE))\\
\Delta_{7}' & = & (({\tt r} = \NIL \PAND \LTRUE) \POR (\NODE({\tt r}, t', {\tt x}) \SPAND \LTRUE) \POR
(t' \neq \NIL \PAND \NODE({\tt r}, t', d') \SPAND \LTRUE))\\
\Delta_{5}' & = & ({\tt r} = t' \PAND \LTRUE)
\end{eqnarray*}
\normalsize
The algorithm stops at $\Delta_{5}'$, since $\Delta_{5} \vdash \Delta_{5}'$, and
calls \KwRefineAbstraction to augment the abstraction.  Our implementation
looks for equalities in each $\Delta'$ that can be used to
strengthen $\Delta$.  In this case, in $\Delta_{10}'$ the heuristic identifies
${\tt d} = {\tt x}$ to strengthen the corresponding $\Delta_{10} = \NODE({\tt
r}, \NIL, d') \PAND {\tt t} = {\tt r} \PAND {\tt res} = 0 \PAND {\tt d} = d'$.
Thus the heuristic selects the variable ${\tt x}$ to augment the abstraction's
multiset.

In the unrefined analysis, any predicate $\LSEG(\_,\_, S)$ will be abstracted
to $\LSEG(\_,\_,\emptyset)$ (equivalent to $\LSEG(\_,\_)$).  Adding {\tt x} to
the multiset means that predicates of the form $\LSEG(\_,\_,\{{\tt x}\})$ will be
protected from abstraction.
We restart the forward analysis from $t_5$. This time
the error is avoided, and we obtain the following abstract states:
\begin{eqnarray*}
t_{6}' & = & (l_{5}, \LSEG({\tt r}, \NIL, \{ {\tt x} \})) \ \ldots \\
t_{9}' & = & (l_{9}, \LSEG({\tt r}, \NIL, \{ {\tt x} \}) \PAND {\tt t} = {\tt r} \PAND {\tt res} = 0)
\end{eqnarray*}
Executing from $l_{9}$ to $l_{10}$ gives two possible post-states:
$\NODE({\tt r}, r', d') \SPAND \LSEG(r', \NIL, \{ {\tt x} \}) \PAND {\tt t} = {\tt r} \PAND {\tt res} = 0 \PAND {\tt d} = d'$ and
$\NODE({\tt r}, r', {\tt x}) \SPAND \LSEG(r', \NIL, \emptyset) \PAND {\tt t} = {\tt r} \PAND {\tt res} = 0 \PAND {\tt d} = {\tt x}$.
In fact, this refined abstraction suffices to prove the absence of errors on
all paths, which completes the analysis. (Other examples may need multiple
refinements)
\vspace{-0.25em} 

\section{Example Multiset-Parametric Analyses}
\label{sec:list-multiset-refinement}
\vspace{-0.25em} 

We describe in detail linked lists with value refinement
and sketch two other multiset families: linked
lists with address refinement, and sorted linked lists with value refinement.
Details for the latter two can be found in
Appendix~\ref{app:list-multiset-refinement}.  All three families are
experimentally evaluated in \S \ref{sec:experiments}.

\label{sec:mls}

Linked lists with value refinement is domain used in our running
example (\S \ref{sec:example}). List segments are instrumented with a multiset
representing the lower bound on the frequency of each variable or
constant. The abstraction function is parameterised by a multiset
controlling which symbols are abstracted. By expanding the multiset,
the preserved frequency bounds are increased, and so the abstraction is
refined.

The domain $\SHmls$ contains spatial predicates $\NODE(\cdot, \cdot, \{ d \})$
and $\LSEG(\cdot, \cdot, S)$ for all $S$ and $d \in S$. Here $x,y$ are
locations, $d$ is a data value, $S$ is a multiset:
\begin{itemize}
\item $\NODE(x, y, \{ d \})$ holds if $x$ points to a node whose
next field contains $y$ and data field contains $d$, i.e.,
$\NODE(x, y, \{ d \}) \, \eqdef \, \POINTSTO{x}{\{ {\sf next}\colon y, {\sf
data}\colon d \}}$.
\item $\LSEG(x, y, S)$ holds if $x$ points to the first node of a
non-empty list segment that ends with $y$, and for each value $d \in
{\rm dom}(S)$, there are at least $S(d)$ nodes that store
$d$.
\end{itemize}
The recursive definition of $\LSEG$ in $\SHmls$ is shown in
Fig.~\ref{fig:mlseg-definition}.  We use these equivalences as folding and
unfolding rules when solving entailment queries in $\SHmls$.

\begin{figure}[t]
\centering
\begin{tabular}{@{}l@{\hspace{0.7em}}rl@{}}
case $S=\emptyset$: & $\LSEG(e,f,\emptyset)\;\eqdef\;$ & $\NODE(e,f,\_) \POR
(\NODE(e,x',\_) \SPAND \LSEG(x',f,\emptyset))$
\\
case $S=\{ d \}$: & $\LSEG(e,f,\{ d \})\;\eqdef\;$ & $\NODE(e,f,\{ d \}) \POR
(\NODE(e,x',\{ d \}) \SPAND \LSEG(x',f,\emptyset))$\\
& & $ \POR (\NODE(e,x',\_) \SPAND \LSEG(x',f,\{ d \}))$\\
case $|S| > 1, d \in S$: & $\LSEG(e,f,S)\;\eqdef\;$ & $\NODE(e,x',\{ d \})
      \SPAND \LSEG(x',f,S \setminus \{ d \}) \POR $\\
 & & $\NODE(e,x',\_) \SPAND \LSEG(x',f,S)$
\end{tabular}
\vspace{-1em} 
\caption{Recursive definition of the $\LSEG$ predicate in domain $\SHmls$.}
\label{fig:mlseg-definition}
\vspace{-0.20em} 
\end{figure}

\begin{figure}[t]
\centering
$$
\begin{array}{r@{\;\;}c@{\;\;}l}
\Delta \PAND x'=e
& \rightsquigarrow_T^{\sf mls} &
\Delta[e/x']
\\[0.7em]
\Delta \SPAND \sigma(x',e,\_) & \rightsquigarrow_T^{\sf mls} & \Delta \SPAND \LTRUE
\hfill \textrm{if }  x' \notin \EVars(\Delta)
\\[0.7em]
\Delta \SPAND \sigma_1(x',y',\_) \SPAND \sigma_2(y',x',\_) & \rightsquigarrow_T^{\sf mls} & \Delta \SPAND \LTRUE
\hfill \textrm{if } x',y' \notin \EVars(\Delta)
\\[0.7em]
\Delta \SPAND \sigma_1(e_1,x',S_1) \SPAND \sigma_2(x',e_2,S_2) & \rightsquigarrow_T^{\sf mls} & \Delta \SPAND \LSEG(e_1,\NIL, \Project_T(S_1 \cup S_2, \Pi)) \\[0.2em]
&& \hfill \textrm{if } x' \notin \EVars(\Delta,e_1,e_2) \wedge \Delta \vdash e_2 = \NIL
\\[0.7em]
\left(
\begin{array}{@{}l@{}}
\Delta \SPAND \sigma_1(e_1,x',S_1) \SPAND
\\
\sigma_2(x',e_2,S_2) \SPAND \sigma_3(e_3,f,S_3)
\end{array}
\right)
& \rightsquigarrow_T^{\sf mls} &
\left(
\begin{array}{@{}r@{}}
\LSEG(e_1,e_2, \Project_T(S_1 \cup S_2, \Pi))
\\
\SPAND \Delta \SPAND  \sigma_3(e_3,f, S_3)
\end{array}
\right)
\\[0.9em]
& \multicolumn{2}{r}{
\;\;\;
\textrm{if } x' \notin \EVars(\Delta, e_1, e_2, e_3, f) \wedge \Delta \vdash e_2 = e_3
}
\\[0.7em]
\Delta \SPAND \LSEG(e,f,S) & \rightsquigarrow_T^{\sf mls} & \LSEG(e,f,\Project_T(S, \Pi))
\end{array}
$$
\vspace{-2em} 
\caption{Abstract reduction system $\rightsquigarrow_T^{\sf mls}$
defining the abstraction function $\Abs_T^{\sf mls}$.
(In the rules, $\sigma, \sigma_i$ range over $\{ \NODE, \LSEG \}$.
The pure assumption $\Pi$ is supplied by the analysis.)}
\label{fig:mls-abs-definition}
\vspace{-0.35em} 
\end{figure}

\paragraph{Abstraction.}
Let $T$ be a finite multiset of program variables and constants. In
Fig.~\ref{fig:mls-abs-definition}, we define a parametric reduction
system $\rightsquigarrow_T^{\sf mls}$, which rewrites symbolic heaps
from $\SHmls$ to canonical heaps whose data and multiset values are
congruent to elements of $T$. Except for the final rule, the relation
$\rightsquigarrow_T^{\sf mls}$ resembles the abstraction for plain
linked lists developed by Distefano et al.~\cite[table
2]{DBLP:conf/tacas/DistefanoOY06}.

The final reduction rule replaces every predicate $\LSEG(e,f,S)$ with
the bounded predicate $\LSEG(e,f,\Project_T(S,\Pi))$. The operator
$\Project_T$ extracts the maximal subset of $S$ such that no element
appears more frequently than it does in $T$ (modulo given pure
assumptions $\Pi$).
Let $\sim_{\Pi}$ be the equivalence relation $x \sim_{\Pi} y \eqdef \Pi
\vdash x = y$. Fix a representative for each equivalence class of
$\sim_{\Pi}$, and for a multiset $S$, denote by $S/_{\Pi}$ the multiset
of $\sim_{\Pi}$-representatives where the multiplicity of a
representative $x$ is $\sum_{x \sim_{\Pi} y} S(y)$. Writing $x\cdot n$
for a multiset element $x$ occurring with multiplicity $n$, we define
$\Project_T$ by
$$
\Project_T(S,\Pi) \; \eqdef \ \{ x \cdot n \mid x \cdot k' \in S/_{\Pi}
\PAND \exists d' \cdot m' \in T/_{\Pi} \ldotp \Pi \vdash x = d' \PAND
n=\min(k',m') \}.
$$
As $\rightsquigarrow_T^{\sf mls}$ has no infinite reduction sequences, it gives
rise to an abstraction function $\Abs_T^{\sf mls}$ by exhaustively
applying the rules until none apply.

\begin{lemma}[Finiteness]
If $T$ is finite and there are only finitely many program variables
then the domain
$\CSH_T^{\sf mls} \eqdef \{ \Delta \mid \Delta \not \vdash \LFALSE
\PAND \Delta \not \rightsquigarrow_T^{\sf mls} \}$ is finite.
\end{lemma}

\begin{lemma}[Soundness]
As $\Delta \rightsquigarrow_T^{\sf mls} \Delta'$ implies $\Delta \vdash
\Delta'$, $\Abs_T^{\sf mls} \colon \SH \to \CSH_T^{\sf mls}$ is a sound
abstraction function.
\end{lemma}

\begin{lemma}[Monotonicity]
If $T_1 \subseteq T_2$ then $\Abs_{T_1}^{\sf mls} \preceq
\Abs_{T_2}^{\sf mls}$.
\end{lemma}
\vspace{-0.75em} 

\subsection{Linked Lists with Address Refinement}
\label{sec:ls}
\vspace{-0.25em} 

\begin{wrapfigure}{r}{4.5cm}
\vspace{-1.55cm}
\begin{lstlisting}
  int remove(Node x) {
    ... // (border cases)
    p = hd;  c = p->next;
    while (c!=nil) {
      if (c==x) {
        p->next = c->next;
        return 1;
      }
      p = c; c = p->next;
    }
    return 0;
  }
\end{lstlisting}
\vspace{-1cm}
\end{wrapfigure}
Rather than preserving certain values in the list, we might need to
preserve nodes at particular \emph{addresses}. For example, to remove
a node from a linked list we might use the procedure shown on the
right.  Given pre-condition $\LSEG({\tt r}, {\tt x}) \SPAND \NODE({\tt
x}, n', \_) \SPAND \LSEG(n', \NIL)$ the procedure will return
\lstinline$1$. However, the standard list abstraction will forget the
existence of the node pointed to by {\tt x}, making this impossible to
prove.

To preserve information of this kind, we combine the domain of linked
lists, $\SHrls$, with a multiset refinement that preserves particular
addresses. Because node addresses are unique, the domain contains just
$\LSEG$ and $\NODE$ predicates, rather than predicates instrumented
with multisets. The reduction system $\rightsquigarrow_T^{\sf rls}$
protects addresses in the multiset $T$ from abstraction. As before,
refinement consists of adding new addresses to the multiset.
\vspace{-0.5em} 

\subsection{Sorted Linked Lists with Value Refinement}
\vspace{-0.25em} 

We can apply the idea of value refinement to different basic domains,
allowing us to deal with examples where different data-structure
invariants are needed.
In our third analysis family, we refine on the existence of particular
values in a sorted list \emph{interval}, rather than a simple segment.
The domain $\SHsls$ contains the predicate $\LSEGless$, parameterised
by an interval of the form $[\alpha, \beta \rangle$, which stores the
bounds of the values in the list, and a multiset $S$, which bounds on
the frequency of particular values in the interval.  The abstraction
function $\rightsquigarrow_T^{\sf sls}$ works in a similar way to
$\SHmls$: the operator $\Project_T$ caps the frequency set $S$,
limiting the number of values that are preserved by abstraction.
\vspace{-0.25em} 

\section{Experimental Evaluation}
\label{sec:experiments}
\vspace{-0.25em} 

We implemented
Algorithm \ref{fig:Fixpoint-Abstraction-Refinement}
and abstract domains $\SHmls$, $\SHrls$
and $\SHsls$ in the separation logic tool ${\sf coreStar}$ \cite{BD+11}.
Aside from superficial tweaks,
we used an identical algorithm and $\KwRefineAbstraction$ heuristic for
all of our case studies.
We used client-oriented specifications
\cite{DBLP:conf/vstte/HaydenMHFF12} describing
datastructures from Redis (a key-value store), Azureus (a BitTorrent
client) and FreeRTOS (real time operating system). Table
\ref{fig:experiments} shows results.

\begin{table}[tb]
\vspace{-1em} 
\begin{center}
\begin{tabular}{c@{\;\;\;}lc@{\hspace{0.1cm}}c@{\hspace{0.1cm}}c@{\hspace{0.1cm}}c@{\hspace{0.1cm}}c@{\hspace{0.1cm}}c}
No & Benchmark & Result & Dom & $|T|$ & \#Refn & $|$ART$|$ & \#Quer \\
\hline
\multicolumn{8}{l}{
\textbf{Set}
}\\
1 & ${\sf add}(x)$--$\ast$--${\sf mem}(x)$ & $\checkmark$ & $\SHmls$ & 1 & 1 & 83 & 162 \\
2 & $\ast$--${\sf add}(x)$--$\ast_{\neg {\sf del}}$--${\sf mem}(x)$ & $\checkmark$ & $\SHmls$ & 1 & 1 & 104 & 193 \\
3 & $\ast_{\neg {\sf del}}$--${\sf mem}(x)$--$\ast_{\neg x}$--${\sf mem}(x)$ & $\checkmark$ & $\SHmls$ & 1 & 1 & 165 & 280 \\
4 & $\ast_{{\sf add}(x)}$--${\sf all\_equal\_to\_}x$ & $\infty$ & $\SHmls$  \\
5 & $\ast_{{\sf add}(x)}$--${\sf all\_sorted}$ & $\top$ & $\SHmls$  \\
\hline
\multicolumn{8}{l}{
\textbf{Multiset}
}\\
6 & ${\sf add}(x)$--${\sf add}(x)$--${\sf del}(x)$--${\sf mem}(x)$ & $\checkmark$ & $\SHmls$ & 2 & 1 & 67 & 91 \\
7 & $\ast$--${\sf add}(x)$--$\ast_{\neg {\sf del}}$--${\sf mem}(x)$ & $\checkmark$ & $\SHmls$ & 1 & 1 & 112 & 205 \\
8 & $\ast_{\neg {\sf del}}$--${\sf mem}(x)$--$\ast_{\neg x}$--${\sf mem}(x)$ & $\checkmark$ & $\SHmls$ & 1 & 1 & 171 & 312 \\
9 & $\ast$--${\sf add}(x)$--$\ast_{\neg {\sf del}}$--${\sf add}(x)$--$\ast_{\neg {\sf del}}$--${\sf del}(x)$--${\sf mem}(x)$ & $\checkmark$ & $\SHmls$ & 2 & 2 & 219 & 458 \\
\hline
\multicolumn{8}{l}{
\textbf{Map}
}\\
10 & $\ast$--${\sf put}(k,v)$--$\ast_{\neg k}$--${\sf get}(k)$ & $\checkmark$ & $\SHmls$ & 1 & 1 & 118 & 215 \\
11 & $\ast$--${\sf rem}(k)$--${\sf bound}(k)$ & $\checkmark$ & $\SHmls$ & 1 & 1 & 92 & 168 \\
\hline
\multicolumn{8}{l}{
\textbf{ByteBufferPool}
}\\
12 & Property 1 & $\checkmark$ & $\SHrls$ & 1 &     1 & 154 & 231 \\
13 & Property 2 & $\checkmark$ & $\SHrls$ & 2 (1) & 1 & 189 & 270 \\
14 & Property 3 & $\checkmark$ & $\SHrls$ & 6 (2) & 4 & 316 & 511 \\
\hline
\multicolumn{8}{l}{
\textbf{FreeRTOS list}
}\\
15 & Property 4 & $\checkmark$ & $\SHmls$ & 1 & 1 & 91 & 158 \\
16 & Property 5 & $\checkmark$ & $\SHsls$ & 6 & 5 & 425 & 971
\end{tabular}
\end{center}
\caption{Experimental Results. Benchmarks
verified by the analysis are marked with $\checkmark$,
those where it threw
an error with $\top$ and those where it did not terminate
with $\infty$. \emph{Dom} is the domain used for the analysis,
\emph{$|T|$} is the size of the multiset $T$ after the final refinement
(number in parentheses denotes the size of the minimal sufficient $T$),
\emph{\#Refn} is the no. of refinement steps, \emph{$|$ART$|$} the
no. symbolic states in the final ART, and \emph{\#Quer}
the no. queries sent to the prover.}
\label{fig:experiments}
\vspace{-2em} 
\end{table}

\smallskip
\noindent
\emph{\textbf{Set}, \textbf{Multiset} and \textbf{Map}.}
%
These are synthetic benchmarks based on specifications
for Redis~\cite{DBLP:conf/vstte/HaydenMHFF12}. They
check various aspects of functional correctness---for example, that
following deletion a key is no longer bound in the store. Furthermore,
we check these specifications across dynamic updates which may modify
the data structures involved---for example, by removing duplicate
bindings to optimize for space usage.

The \textbf{Set} and \textbf{Multiset} benchmarks apply operations
${\sf add}$ (add an element), ${\sf del}$ (delete an element) and ${\sf
mem}$ (test for membership) to a list-based set (multiset,
respectively) in the order indicated by the benchmark name. The symbols
$\ast$, $\ast_{\neg {\sf del}}$ and $\ast_{\neg x}$ respectively denote
applying all operations any number of times with any argument, all
operations except ${\sf del}$, and all operations but excluding $x$ as
an argument. For $\textbf{Map}$ benchmarks the operations ${\sf put}$
(insert a key-value pair), ${\sf get}$ (retrieve a value for the given
key), ${\sf rem}$ (remove a key with the associated value) and ${\sf
bound}$ (check if the key is bound) are to a list-based map. For
benchmarks 1,2,6,7,9,10 the goal was to prove that the last operation
returns ${\sf true}$; for benchmark 11 that it returns
${\sf false}$; and, for benchmarks 3 and 8 that the two ${\sf mem}$
operations return the same value. Benchmark 4 illustrates a universal
property that causes our analysis in $\SHmls$ to loop forever by adding
$x$ to $T$ at each refinement step. Benchmark 5 is a universal property
for which our analysis in $\SHmls$ fails to find an inductive invariant due to
the ordering predicate (using $\SHsls$ on the same benchmark
loops forever).
\vspace{-0.25em} 

\smallskip
\noindent
\emph{\textbf{ByteBufferPool}.}
Azureus uses a pool of ByteBuffer objects to store results of network
transfers.  In early versions, free buffers in this pool were
identified by setting the buffer position to a sentinel value.  The
\textbf{ByteBufferPool} benchmarks check properties of this pool.
Property 1 checks that if the pool is full and a buffer is freed, that
just-freed buffer is returned the next time a buffer is requested.
Property 2 checks that if the pool has some number of free buffers,
then no new buffers are allocated when a buffer is requested.  Property
3 checks that if the pool has at least two free buffers, then two
buffer requests can be serviced without allocating new buffers.
\vspace{-0.25em} 

\smallskip
\noindent
\emph{\textbf{FreeRTOS list}.}
This is a sorted cyclic list with a sentinel node, used task management in
the scheduler. The value of the sentinel marks the end of the
list---for instance, on task insertion the list is
traversed to find the right insertion point and the guard for that
iteration is the sentinel value. To check correctness of the shape
after insertion (Property 4) it suffices to remember that the sentinel
value is in the list. To check that tasks are also correctly sorted
according to priorities (Property 5) we need to keep track of list
sortedness and all possible priorities as splitting points.
\vspace{-0.25em} 


\section{Conclusions and Limitations}
\label{sec:abs_ref_properties}
\vspace{-0.25em} 

We have presented a CEGAR-like abstraction refinement scheme for separation
logic analyses, aimed at refining existential properties of programs, in
which we want to track some elements of a data structure more precisely
than others.

Our prototype tool is built on ${\sf coreStar}$~\cite{BD+11}, and we expect our approach
would combine well with other separation logic tools,
e.g.~\cite{DBLP:journals/jacm/CalcagnoDOY11,DBLP:conf/cav/BerdineCI11}.
In particular, abduction is known to work well in an inter-procedural
setting~\cite{DBLP:journals/jacm/CalcagnoDOY11} and we thus believe our
approach could be made inter-procedural without substantial further research.

Minimizing incompleteness is more challenging, as
without further assumptions
Algorithm~\ref{fig:Fixpoint-Abstraction-Refinement} might diverge, or
fail to recognize a spurious coun\-ter-example as infeasible. If
the forward transfer function is \emph{exact} (i.e., returns the
strongest post-condition) and the backward transfer function is
\emph{precise} (i.e., for any $c$ and $\Delta$, $\llbracket c
\rrbracket (\llbracket c \rrbracket^{\leftarrow} (\Delta)) \vdash
\Delta$) then the algorithm makes progress relative to the refinement
heuristic. Intuitively, if $\KwRefineAbstraction$ always picks a symbol
such that the refined abstraction rules out the spurious
counter-example, then that counter-example will never reappear in
subsequent iterations. However, we are skeptical that our current
heuristic satisfies this condition. For a more formal discussion, see
Appendix~\ref{sec:progress_completeness}.

Note that Berdine et al.~\cite{berdine-cav2012} similarly do not establish
progress for their analysis.
Predicate abstraction techniques that do not
\textit{a priori} fix the set of predicates have the same issue, as do
interpolation-based procedures that do not constrain the language of acceptable
interpolants.  In both cases, the restrictions that ensure termination also
limit the set of programs that can be proved correct.

\stephencomment{Do we need citations here?}

\bibliographystyle{abbrv}
\bibliography{references}

\newpage
\appendix

\input{tech_apdx}

\end{document}

%% file: tech_apdx.tex
\section{Relative Progress and Completeness}
\label{sec:progress_completeness}

Without further assumptions, the abstraction refinement algorithm might
diverge, or report a spurious counter-example which is in fact not
feasible.  The following idealised assumptions suffice to ensure
progress and completeness (we are skeptical that condition (c) holds
for our current realisation of the analysis---see below).
\begin{itemize}
\item[(a)] The forward transfer function is exact (i.e.,
    $\llbracket \cdot \rrbracket$-image is the strongest
    post-condition in the given abstract domain).

\item[(b)] The backward transfer function is precise (so we are
    able to identify spurious counter-examples). Formally, for any
    $c$ and $\Delta$, we have $\llbracket c \rrbracket (\llbracket
    c \rrbracket^{\leftarrow} (\Delta)) \vdash \Delta$.

\item[(c)] When called with a $(\Path_{\sf fwd}, \Path_{\sf bwd})$-pair
    of the counter-example and the path sufficient to avoid the error, the
    procedure call \KwRefineAbstraction{$\Path_{\sf fwd}, \Path_{\sf bwd}$}
    picks symbols $A$ for augmenting $S$ such that the spurious
    counter-example ending with $\Path_{\sf fwd}$ is eliminated by the
    abstraction $\Abs_{S \cup A}$.
\end{itemize}
Alg.~\ref{fig:Fixpoint-Abstraction-Refinement} then makes progress
by ensuring that a counter-example, once eliminated, remains eliminated in
all subsequent iterations.

\begin{theorem}[Relative progress]
Let $\gamma_j$ be the counter-example processed in the $j$-th refinement
step. Then for all $j \geq 1$, $|\gamma_j| \leq |\gamma_{j+1}|$, where
$| \cdot |$ denotes the length of the counter-example. In addition, if
$\gamma_j$ is processed with value $k$ in the while-loop on line
\ref{fig:Fixpoint:main-loop} of
Alg.~\ref{fig:Fixpoint-Abstraction-Refinement} then the program being
analysed has no counter-examples of length less than $k$.
\end{theorem}

\begin{proof}
Let $S(j)$ denote the multiset from the $j$-th iteration of \KwRefine. 
Since $\Abs_{S(j)} \preceq \Abs_{S(j+1)}$, no new counter-examples can 
appear in the part of the ART that is recomputed in the $(j+1)$-th step 
(invariants computed in $\CSH_{S(j+1)}$ will be at least as strong as 
those in $\CSH_{S(j)}$). Since (c) guarantees that the previous 
counter-example has been eliminated, if a new counter-example is found 
then the corresponding value of $k$ in the while-loop will be either 
the same as in the $j$-th step or larger.
\qed
\end{proof}

\begin{theorem}[Relative completeness]
If the safety property is implied by an inductive invariant expressible
in $\CSH_S$ for some finite multiset $S$ and assuming that those
elements would eventually be selected from counter-examples by
\KwRefineAbstraction then
Alg.~\ref{fig:Fixpoint-Abstraction-Refinement} terminates without
throwing an error.
\end{theorem}

\begin{proof}
Since Alg.~\ref{fig:Fixpoint-Abstraction-Refinement} proceeds in a
breadth-first fashion and counter-examples to safety properties are
finite, all counter-examples leading to picking elements of $S$ will
eventually be processed, enabling
Alg.~\ref{fig:Fixpoint-Abstraction-Refinement} to compute an invariant
in $\CSH_{S'}$ for some $S' \supseteq S$.
\qed
\end{proof}

\noindent
Assumptions (a) and (b) can be satisfied (although for implementation
efficiency we may choose not to).  Assumption (c) is more problematic.

\paragraph{Forward transfer.}
Without exactness, a spurious counter-example may never be eliminated,
because our analysis refines only the abstraction function. Since
separation logic analyses effectively calculate strongest
post-conditions,\footnote{modulo deallocation---although even for that
case the forward transfer is tight in actual implementations.} we in
fact have exact forward transfer, meaning spurious counter-examples can
always be eliminated.

\mikecomment{Explanation for the deallocation comment?}

\paragraph{Backward transfer.}
In our analysis abduction is performed on finite unfoldings of predicates,
modulo an arbitrary frame, fixed along the counter-example.  As a result,
counter-examples are always expressed as data-structures of
a particular size (rather than e.g. general lists which could be of any size).
This means that counter-examples can be expressed
in the points-to fragment of separation logic, in which optimal solutions are
possible~\cite{DBLP:journals/jacm/CalcagnoDOY11}. Thus in principle we
can satisfy (b) and make backward transfer precise. However, such a
complete abductive inference is of exponential complexity since it has
to consider all aliasing possibilities. In our implementation, we use a
polynomial heuristic algorithm (similar to
\cite{DBLP:journals/jacm/CalcagnoDOY11}) which may miss some solutions,
but in practice has roughly the same cost as frame inference.


\paragraph{Selecting symbols.}
Due to its heuristic nature, it is unlikely that our implementation of
\KwRefineAbstraction satisfies assumption (c). Furthermore, we are
unsure whether it is generally possible to construct
\KwRefineAbstraction that would satisfy (c) for an arbitrary parametric
domain. While at least in principle we could employ a trivial heuristic
which enumerates all multisets of symbols, that would be impractical.
The problem of picking symbols which are certain to eliminate a
particular counter-example seems uncomfortably close to selecting
predicates for predicate abstraction sufficient to prove a given
property. Many effective heuristics used in this area are incomplete
(in that they may fail to find an adequate set of predicates when one
exists), and there has been only a limited progress in characterising
complete methods.\footnote{See Ranjit Jhala, Kenneth L. McMillan. A
Practical and Complete Approach to Predicate Refinement. In
\emph{TACAS}, 2006, for an instance of such complete predicate
refinement method (for difference bound arithmetic over the
rationals).} Unfortunately, all such complete predicate refinement
methods rely on interpolation, a luxury which we do not (yet) have in
separation logic. More work is needed to understand the intrinsic
complexity of ways for doing refinement in separation logic analyses
such as the one proposed in this paper in relation to the logical
properties of separation logic domains.

\section{Details of Other Multiset-Parametric Domains}
\label{app:list-multiset-refinement}

Here we give detailed definitions of the two analysis families that we
sketched in \S\ref{sec:list-multiset-refinement}.

\subsection{Linked Lists with Address Refinement}

This analysis allows refinement on protecting particular addresses, rather than
values.  We work with the domain of linked lists, which we denote $\SHrls$,
built from plain spatial predicates $\NODE$ and $\LSEG$.

\begin{figure}[t]
\centering
$$
\begin{array}{r@{\;\;}c@{\;\;}l}
\Delta \SPAND \sigma_1(e_1,x') \SPAND \sigma_2(x',e_2) & \rightsquigarrow_T^{\sf rls} & \Delta \SPAND \LSEG(e_1,\NIL) \\[0.2em]
\multicolumn{3}{r}{
\;\;\;\;\;\;\;\;\;\;\;\;\;\;\;\;\;\;\;\;\;\;\;\;\;\;\;\;\;\;
\textrm{if } x' \notin \EVars(\Delta,e_1,e_2) \wedge \Delta \vdash e_2 = \NIL \wedge \forall t \in T \,.\, \Delta \nvdash e_1 = t
}
\\[0.7em]
\Delta \SPAND \sigma_1(e_1,x') \SPAND \sigma_2(x',e_2) \SPAND \sigma_3(e_3,f) & \rightsquigarrow_T^{\sf rls}
& \Delta \SPAND \LSEG(e_1,e_2) \SPAND \sigma_3(e_3,f) \\[0.2em]
\multicolumn{3}{r}{
\;\;\;\;\;\;\;\;\;\;\;\;\;\;\;\;\;\;\;\;\;\;\;\;\;\;\;\;\;\;
\textrm{if } x' \notin \EVars(\Delta, e_1, e_2, e_3, f) \wedge \Delta \vdash e_2 = e_3 \wedge \forall t \in T \,.\, \Delta \nvdash e_1 = t
}
\end{array}
$$
\vspace{-2em} 
\caption{Abstract reduction system $\rightsquigarrow_T^{\sf rls}$
defining the abstraction function $\Abs_T^{\sf rls}$. First three rules
(not shown) are the same as in Fig.~\ref{fig:mls-abs-definition}.
In the shown rules, $\sigma, \sigma_i$ range over $\{ \NODE, \LSEG \}$
and the data field is elided.}
\label{fig:ls-abs-definition}
\end{figure}

Our abstraction works similarly to the abstraction for plain
linked lists \cite{DBLP:conf/tacas/DistefanoOY06} except that it can be
refined to preserve nodes at particular addresses.
Fig.~\ref{fig:ls-abs-definition} shows rewrite rules realising the
abstraction $\Abs_T^{\sf rls}$. The rules are guarded by a finite set of
terms $T$ representing locations---each rule is enabled only if
the spatial object triggering the rule is not among the
locations in $T$.


\begin{lemma}
$\CSH_T^{\sf rls} \eqdef \{ \Delta \mid \Delta \not \vdash \LFALSE
\PAND \Delta \not \rightsquigarrow_T^{\sf rls} \}$ is finite and
$\Abs_T^{\sf rls} \colon \SH \to \CSH_T^{\sf rls}$ is a sound
abstraction.
\end{lemma}

\begin{lemma}
If $T_1 \subseteq T_2$ then $\Abs_{T_1}^{\sf rls} \preceq
\Abs_{T_2}^{\sf rls}$.
\end{lemma}

\subsection{Sorted Linked Lists with Value Refinement}

Lastly, we present an analysis that works in the domain of sorted linked lists.
Our abstraction can be refined to preserve particular values in the list, as
with the analysis described in \S\ref{sec:mls}. However, the domain consists of
ordered lists segments.

\begin{figure}[t]
\centering
\begin{tabular}{@{}rl@{}}
$S=\emptyset$: \hfill {}
\\
$\LSEGless(e,f,[\alpha, \beta \rangle, \emptyset) \;\eqdef\; $ & $\alpha \leq d' < \beta \PAND (\NODE(e,f,\{ d' \}) \POR$\\
& $\NODE(e,x',d') \SPAND \LSEGless(x',f,[d', \beta \rangle,\emptyset))$\\
$S=\{ d \}$: \hfill {}
\\
$\LSEGless(e,f,[\alpha, \beta \rangle,\{ d \})\;\eqdef\;$ & $\NODE(e,f,\{ d \}) \POR$\\
& $d = \alpha \PAND \NODE(e,x',\{ d \}) \SPAND \LSEGless(x',f,[d, \beta \rangle,\emptyset) \POR$\\
& $d \neq \alpha \PAND \alpha \leq d' < \beta \PAND \NODE(e,x',d') \SPAND \LSEGless(x',f,[d', \beta \rangle,\{ d \})$\\
$|S| > 1, d \in S$: \hfill {}
\\
$\LSEGless(e,f,[\alpha, \beta \rangle, S)\;\eqdef\; $ & $d = \alpha \PAND \NODE(e,x',\{ d \}) \SPAND \LSEGless(x',f,[d, \beta \rangle,S \setminus \{ d \}) \POR $\\
& $d \neq \alpha \PAND \alpha \leq d' < \beta \PAND \NODE(e,x',d') \SPAND \LSEGless(x',f,[d', \beta \rangle,S)$
\end{tabular}
\vspace{-1em} 
\caption{Recursive definition of the $\LSEGless$ predicate in domain $\SHsls$.}
\label{fig:slseg-definition}
\end{figure}

\paragraph{Domain.}
The predicate $\LSEGless(x, y, [\alpha, \beta
\rangle, S)$ holds if $x$ points to a \emph{sorted} non-empty list
segment ending with $y$ whose data values are all greater than or equal
to $\alpha$ and less than $\beta$, and for each $d \in {\rm dom}(S)$,
there are at least $S(d)$ nodes in the list with value $d$. Parameters
$\alpha, \beta$ and $S$ satisfy the invariant $I
\colon \forall d \in {\rm dom}(S) \, . \, \alpha \leq d < \beta$.
Sorted lists can be split according to the following rule:
$$\LSEGless(e, f, [\alpha, \beta \rangle, S) = \LSEGless(e, x', [\alpha, \gamma \rangle, S \cap [\alpha, \gamma \rangle)
\SPAND \LSEGless(x', f, [\gamma, \beta \rangle, S \cap [\gamma, \beta \rangle).$$
Folding/unfolding rules for exposing/hiding are similar to the rules
for $\LSEG$ (Fig.~\ref{fig:mlseg-definition}), but in addition
keep track of the involved inequalities. New rules for $\LSEGless$ are
shown in Fig.~\ref{fig:slseg-definition}. Note that each rule
maintains the invariant $I$.

\begin{figure}[t]
\centering
$$
\begin{array}{r@{\;\;}c@{\;\;}l}
\Delta \SPAND \NODE(e_1,x',\{ d_1 \}) \SPAND \NODE(x',e_2,\{ d_2 \}) & \rightsquigarrow_T^{\sf sls} & \\[0.2em]
\multicolumn{3}{r}{
\;\;\;\;\;\;\;\;\;\;\;\;\;\;\;\;\;\;\;\;\;\;\;\;\;\;\;\;\;\;\;\;\;\;\;\;\;\;\;\;\;\;\;\;\;\;\;\;\;\;
\Delta \SPAND \LSEGless(e_1,\NIL,[d_1,d_2 + 1 \rangle,\Project_T(\{ d_1, d_2\}, \Pi))}
\\[0.2em]
\multicolumn{3}{r}{
\;\;\;\;\;\;\;\;\;\;\;\;\;\;\;\;\;\;\;\;\;\;\;\;\;\;\;\;\;\;\;\;\;\;\;\;\;\;\;\;\;\;\;\;\;\;\;\;\;\;
\textrm{if } x' \notin \EVars(\Delta,e_1,e_2) \wedge \Delta \vdash e_2 = \NIL
}
\\[0.7em]
\Delta \SPAND \LSEGless(e_1,x',[\alpha_1, \beta_1 \rangle,S_1) \SPAND \LSEGless(x',e_2,[\alpha_2, \beta_2 \rangle,S_2) & \rightsquigarrow_T^{\sf sls} & \\[0.2em]
\multicolumn{3}{r}{
\;\;\;\;\;\;\;\;\;\;\;\;\;\;\;\;\;\;\;\;\;\;\;\;\;\;\;\;\;\;\;\;\;\;\;\;\;\;\;\;\;\;\;\;\;\;\;\;\;\;
\Delta \SPAND \LSEGless(e_1,\NIL,[\alpha_1, \beta_2 \rangle, \Project_T(S_1 \cup S_2, \Pi))
}
\\[0.2em]
\multicolumn{3}{r}{
\;\;\;\;\;\;\;\;\;\;\;\;\;\;\;\;\;\;\;\;\;\;\;\;\;\;\;\;\;\;\;\;\;\;\;\;\;\;\;\;\;\;\;\;\;\;\;\;\;\;
\textrm{if } x' \notin \EVars(\Delta,e_1,e_2) \wedge \Delta \vdash e_2 = \NIL \wedge \beta_1 \leq \alpha_2
}
\end{array}
$$
\vspace{-2em} 
\caption{Selected rules of the abstract reduction system $\rightsquigarrow_T^{\sf sls}$
defining the abstraction function $\Abs_T^{\sf sls}$.}
\label{fig:sls-abs-definition}
\end{figure}

\paragraph{Abstraction.}
In the abstraction, we proceed similarly as in
Fig.~\ref{fig:mls-abs-definition} but also maintain the invariant
$I$. Fig.~\ref{fig:sls-abs-definition} shows rewrite rules
corresponding to the fourth rule of Fig.~\ref{fig:mls-abs-definition}
for $\sigma_1 = \sigma_2 = \NODE$ and $\sigma_1 = \sigma_2 =
\LSEGless$. The rest of the cases for
$\sigma_i$ are analogous to the fourth rule, and the fifth rule of
Fig.~\ref{fig:mls-abs-definition}.
The resulting abstraction $\Abs_T^{\sf sls}$ satisfies the following lemmas:

\begin{lemma}
For $\CSH_T^{\sf rss} \eqdef \{ \Delta \mid \Delta \not \vdash \LFALSE
\PAND \Delta \not \rightsquigarrow_T^{\sf sls} \}$, $\Abs_T^{\sf sls}
\colon \SH \to \CSH_T^{\sf rls}$ is a sound abstraction. If the domain
of values is finite then $\CSH_T^{\sf rss}$ is also finite.
\end{lemma}

\begin{lemma}
If $T_1 \subseteq T_2$ then $\Abs_{T_1}^{\sf sls} \preceq
\Abs_{T_2}^{\sf sls}$.
\end{lemma}
For infinite value domains, the set $\CSH_T^{\sf rss}$ is infinite
since we have infinite ascending chains of intervals as parameters to
$\LSEGless$. We could recover convergence in such cases by using
widening on the interval domain \cite{DBLP:conf/popl/CousotC77}.

%% file: paper.bbl
\begin{thebibliography}{10}

\bibitem{DBLP:conf/spin/BallR01}
T.~Ball and S.~K. Rajamani.
\newblock Automatically validating temporal safety properties of interfaces.
\newblock In {\em SPIN}, 2001.

\bibitem{DBLP:conf/aplas/BerdineCO05}
J.~Berdine, C.~Calcagno, and P.~W. O'Hearn.
\newblock Symbolic execution with separation logic.
\newblock In {\em APLAS}, 2005.

\bibitem{DBLP:conf/cav/BerdineCI11}
J.~Berdine, B.~Cook, and S.~Ishtiaq.
\newblock Slayer: Memory safety for systems-level code.
\newblock In {\em CAV}, 2011.

\bibitem{berdine-cav2012}
J.~Berdine, A.~Cox, S.~Ishtiaq, and C.~M. Wintersteiger.
\newblock Diagnosing abstraction failure for separation logic-based analyses.
\newblock In {\em CAV}, 2012.

\bibitem{DBLP:conf/cav/BeyerHT06}
D.~Beyer, T.~A. Henzinger, and G.~Th{\'e}oduloz.
\newblock Lazy shape analysis.
\newblock In {\em CAV}, 2006.

\bibitem{BD+11}
M.~Botin\v{c}an, D.~Distefano, M.~Dodds, R.~Grigore,
  D.~Naud\v{z}i\={u}nien\.{e}, and M.~Parkinson.
\newblock {{\sf coreStar}: The Core of {\sf jStar}}.
\newblock In {\em Boogie}, 2011.

\bibitem{DBLP:journals/jacm/CalcagnoDOY11}
C.~Calcagno, D.~Distefano, P.~W. O'Hearn, and H.~Yang.
\newblock Compositional shape analysis by means of bi-abduction.
\newblock {\em J. ACM}, 58(6), 2011.

\bibitem{DBLP:conf/icse/ChakiCGJV03}
S.~Chaki, E.~M. Clarke, A.~Groce, S.~Jha, and H.~Veith.
\newblock Modular verification of software components in {C}.
\newblock In {\em ICSE}, 2003.

\bibitem{DBLP:conf/sas/ChangRN07}
B.-Y.~E. Chang, X.~Rival, and G.~C. Necula.
\newblock Shape analysis with structural invariant checkers.
\newblock In {\em SAS}, pages 384--401, 2007.

\bibitem{Chin2012bag}
W.-N. Chin, C.~David, H.~H. Nguyen, and S.~Qin.
\newblock Automated verification of shape, size and bag properties via
  user-defined predicates in separation logic.
\newblock In {\em Science of Computer Programming}, volume 77:9, 2012.

\bibitem{DBLP:conf/cav/ClarkeGJLV00}
E.~M. Clarke, O.~Grumberg, S.~Jha, Y.~Lu, and H.~Veith.
\newblock Counterexample-guided abstraction refinement.
\newblock In {\em CAV}, 2000.

\bibitem{DBLP:conf/popl/CousotC77}
P.~Cousot and R.~Cousot.
\newblock Abstract interpretation: A unified lattice model for static analysis
  of programs by construction or approximation of fixpoints.
\newblock In {\em POPL}, 1977.

\bibitem{DBLP:conf/tacas/DistefanoOY06}
D.~Distefano, P.~W. O'Hearn, and H.~Yang.
\newblock A local shape analysis based on separation logic.
\newblock In {\em TACAS}, 2006.

\bibitem{DBLP:conf/oopsla/DistefanoP08}
D.~Distefano and M.~J. Parkinson.
\newblock {jStar}: towards practical verification for {Java}.
\newblock In {\em OOPSLA}, 2008.

\bibitem{DBLP:conf/cav/GrafS97}
S.~Graf and H.~Sa\"{\i}di.
\newblock Construction of abstract state graphs with {PVS}.
\newblock In {\em CAV}, 1997.

\bibitem{DBLP:conf/tacas/GulavaniCNR08}
B.~S. Gulavani, S.~Chakraborty, A.~V. Nori, and S.~K. Rajamani.
\newblock Automatically refining abstract interpretations.
\newblock In {\em TACAS}, 2008.

\bibitem{DBLP:conf/tacas/GulavaniR06}
B.~S. Gulavani and S.~K. Rajamani.
\newblock Counterexample driven refinement for abstract interpretation.
\newblock In {\em TACAS}, 2006.

\bibitem{DBLP:conf/vstte/HaydenMHFF12}
C.~M. Hayden, S.~Magill, M.~Hicks, N.~Foster, and J.~S. Foster.
\newblock Specifying and verifying the correctness of dynamic software updates.
\newblock In {\em VSTTE}, 2012.

\bibitem{DBLP:conf/popl/HenzingerJMS02}
T.~A. Henzinger, R.~Jhala, R.~Majumdar, and G.~Sutre.
\newblock Lazy abstraction.
\newblock In {\em POPL}, 2002.

\bibitem{Kurshan94}
R.~P. Kurshan.
\newblock {\em Computer-aided verification of coordinating processes: the
  automata-theoretic approach}.
\newblock Princeton University Press, 1994.

\bibitem{DBLP:conf/cav/LoginovRS05}
A.~Loginov, T.~W. Reps, and S.~Sagiv.
\newblock Abstraction refinement via inductive learning.
\newblock In {\em CAV}, 2005.

\bibitem{DBLP:conf/cav/McMillan06}
K.~L. McMillan.
\newblock Lazy abstraction with interpolants.
\newblock In {\em CAV}, 2006.

\bibitem{DBLP:conf/popl/NaikYCS12}
M.~Naik, H.~Yang, G.~Castelnuovo, and M.~Sagiv.
\newblock Abstractions from tests.
\newblock In {\em POPL}, 2012.

\bibitem{DBLP:journals/toplas/SagivRW02}
S.~Sagiv, T.~W. Reps, and R.~Wilhelm.
\newblock Parametric shape analysis via 3-valued logic.
\newblock {\em TPLS}, 24(3), 2002.

\bibitem{DBLP:conf/vmcai/Vafeiadis09}
V.~Vafeiadis.
\newblock Shape-value abstraction for verifying linearizability.
\newblock In {\em VMCAI}, 2009.

\bibitem{DBLP:conf/cav/YangLBCCDO08}
H.~Yang, O.~Lee, J.~Berdine, C.~Calcagno, B.~Cook, D.~Distefano, and P.~W.
  O'Hearn.
\newblock Scalable shape analysis for systems code.
\newblock In {\em CAV}, 2008.

\end{thebibliography}
